\pdfoutput=1
\documentclass[11pt]{article}

\usepackage[letterpaper,margin=1in]{geometry}
\usepackage{graphicx}
\usepackage{amsmath}
\usepackage{amssymb}
\usepackage{amsthm}
\usepackage{xcolor}
\usepackage{bm}
\usepackage[numbers,sort&compress]{natbib}
\usepackage{hyperref}

\graphicspath{{./}}

\hypersetup{
    colorlinks=true,
    linkcolor=blue,
    citecolor=blue,
    urlcolor=blue
}

\theoremstyle{plain}
\newtheorem{theorem}{Theorem}
\newtheorem{proposition}[theorem]{Proposition}

\theoremstyle{definition}

\theoremstyle{remark}

\definecolor{mathred}{RGB}{255, 0, 0}

\numberwithin{equation}{section}

\begin{document}

\title{Pi-Change: A Prior-Informed Multiple Change Point Detection Algorithm}

\author{
Jonathon Jacobs\\
Department of Biostatistics\\
Virginia Commonwealth University\\
Richmond, VA, USA
\and
Shanshan Chen\thanks{Corresponding author: \href{mailto:schen3@vcu.edu}{schen3@vcu.edu}}\\
Department of Biostatistics\\
Virginia Commonwealth University\\
Richmond, VA, USA
}

\date{}

\maketitle

\abstract{Statistical change point (CP) detection methods typically rely on likelihood-based inference and ignore contextual information about plausible CP locations beyond the observed sequence. Although informative priors provide a natural way to incorporate such information, general and computationally efficient methods for doing so are lacking, especially for multiple CP detection. To address this gap, we propose a prior-informed CP detection algorithm (Pi-Change) that incorporates prior information on CP locations through a time-varying penalty term. We prove that the proposed penalty can be embedded in the Pruned Exact Linear Time framework while preserving the dynamic programming recursion and pruning rule required for efficient multiple CP detection. Across simulation studies and three time-series applications, Pi-Change discourages spurious CPs unsupported by prior information, remains robust to prior misspecification, and improves detection accuracy. More broadly, Pi-Change extends multiple CP detection beyond purely data-driven fitting by incorporating partial prior knowledge in a computationally efficient and interpretable way. It is particularly useful when CPs arise from heterogeneous mechanisms or are associated with known external events, helping quantify the delay between an event and the resulting structural change.}


\section{Introduction}
Long time series are increasingly prevalent in science as a result of automated collection techniques \citep{jensen2017time}. These time series likely have multiple change points (CPs) when the data generation mechanism changes. Identifying these CPs and interpreting their meaning is a central goal of time-series analysis, creating a need for multiple CP detection (MCD) methods that are both accurate and scalable.

Existing MCD methods broadly fall into two classes. One class resorts to the "divide and conquer" approach, by first partition the long sequences to shorter windows and locating a single CP within each window. This strategy is attractive because it simplifies estimation by reducing the MCD problem to a squence of single-CP detection tasks, and it underlies methods such as binary segmentation and application-specific windowing procedures. Examples of this group include CircaCP \citep{chen2024validating}, an wearable actigraphy segmentation method which produces windows based on the circadian rhythm, and Binary segmentation, a more general algorithm which produces windows that assume the existence of a CP iteratively \citep{sen1975tests}. A second class treats the number and locations of CPs as a global optimisation problem under a segmentwise cost function. Methods in this class, including segment neighborhood \citep{braun1998statistical}, optimal partitioning \citep{jackson2005algorithm}, and pruned exact linear time (PELT) \citep{killick2012optimal}, do not require the number of CPs to be fixed in advance and are often preferred when scalable global estimation is desired.

Despite these advances, current CP detection methods rely almost exclusively on the observed sequence itself. In many applications, however, historical or contextual information about where CPs are likely to occur is not only available, but also scientifically important, because it reflects distinct underlying data-generating mechanisms and helps distinguish scientifically meaningful changes from incidental ones. In actigraphy , for example, circadian rhythm provides such external information about likely sleep- and wake-onset times \citep{chen2024validating}. In historical, economic, or environmental series, known external events may suggest periods in which structural change is more plausible. Methods that treat all candidate CPs as exchangeable may therefore identify changes that are statistically admissible but scientifically uninformative.

Current methods have limited ability to incorporate such information. Bayesian approaches provide a natural mechanism for encoding prior information on CP locations, and a few parametric Bayesian methods \citep{fearnhead2006exact, adams2007bayesian}, or MCMC-based methods have been proposed in simpler settings \citep{stephens1994bayesian, dehning2020inferring}. However, these methods are often computationally burdensome and do not provide a general, scalable solution for MCD. Window-based methods, such as CircaCP, can encode prior information indirectly through the choice of windows, but this representation is rigid and does not quantify uncertainty in the prior information. More importantly, neither approach can be readily integrated into a global optimization framework that retains both flexibility and computational efficiency. Consequently, existing methods do not provide a principled way to let external information guide CP estimation while controlling its influence on the final segmentation.

We thus propose Pi-Change, a prior-informed MCD algorithm that incorporates prior information on likely CP locations through a time-varying penalty. The penalty is constructed so that its minima reflect locations favored by prior information and can be specified to approximate prior distributions on multiple CP locations, thereby admitting a Bayesian interpretation. We further show that the resulting criterion can be embedded within the PELT framework while preserving the recursion and pruning rule required for efficient estimation. 

The rest of the paper is organized as follows. Section 2 details the Pi-Change algorithm, by introducing the construction of the time-varying penalty and its valid integration with PELT. Section 3 presents simulation studies. Section 4 applies Pi-Change to three real-data examples: combat deaths from 1816 to 2007, oil-price volatility from 2000 to 2009, and workplace mobility after the onset of the COVID-19 pandemic. Section 5 concludes with discussion, limitations, and future directions.

\section{Methods}
In this section, we motivate our approach from a Bayesian perspective, construct time-varying penalty functions inspired by Bayesian priors, and then show that these prior-like penalties can be validly incorporated into PELT while preserving its computational efficiency.

\textbf{Problem Setup:} For a time series $ Y = [y_1, y_2,\ldots,y_N]$, let $C(y_{s:t})$ denote the cost of fitting a segment from index $s$ to $t$, and let $\boldsymbol{\tau}=\{\tau_0=0<\tau_1<\cdots<\tau_m<\tau_{m+1}=N\}$ denote a set of CPs. A MCD problem seeks to identify the set of CPs $\boldsymbol{\tau}$ that globally minimize the sum of cost functions:
\begin{equation}
\label{equ: MCD min}
    \arg\min_{\boldsymbol{\tau}} \sum_{i = 1}^{m + 1} \Big[C(y_{((\tau_{i - 1}) + 1): \tau_i}) \Big],
\end{equation}

\subsection{Bayesian analysis for MCD problems}
\label{sec: CP priors}
Bayesian analysis provides a natural way to encode prior information on CP locations through a prior $\pi(\boldsymbol{\tau})$. In the multiple CP setting, this prior can be decomposed into a prior on the number of CPs, $\pi_1(|\boldsymbol{\tau}|)$, and a prior on their locations conditional on that number, $\pi_2(\boldsymbol{\tau}\mid |\boldsymbol{\tau}|)$, i.e., $ \pi(\boldsymbol{\tau}) = \pi_1(|\boldsymbol{\tau}|)\pi_2(\boldsymbol{\tau}\mid |\boldsymbol{\tau}|)$. Assuming that the prior on CP locations is specified independently of the remaining model parameters, the posterior distribution of $\boldsymbol{\tau}$ satisfies
\begin{equation}
    P(\boldsymbol{\tau}\mid Y) \propto P(Y\mid \boldsymbol{\tau})\,\pi_2(\boldsymbol{\tau}\mid |\boldsymbol{\tau}|)\,\pi_1(|\boldsymbol{\tau}|),
\end{equation}
where $P(Y\mid \boldsymbol{\tau})$ is the marginal likelihood under CP configuration $\boldsymbol{\tau}$.

To illustrate this posterior, we consider the single CP problem, where the task is to identify a single true CP $k$. Let $f(Y|\boldsymbol{\Theta})$ be a given data distribution, where $\boldsymbol{\Theta}$ is the parameter vector of the distribution. Suppose that $\pi(k) = \pi_2(\boldsymbol{\tau}\big||\boldsymbol{\tau}| = 1)$ is a proper prior on the CP location, and $\pi(\boldsymbol{\Theta}|\boldsymbol{\alpha})$ is a prior on the parameters $\boldsymbol{\Theta}$ given hyperparameters $\boldsymbol{\alpha}$. Then the posterior distribution of $\boldsymbol{\tau}$ is 

\begin{equation}
    p(k|Y, \boldsymbol{\Theta}, \boldsymbol{\alpha})
    = \frac{p(k, Y, \boldsymbol{\Theta}, \boldsymbol{\alpha})}{p(Y, \boldsymbol{\Theta}, \boldsymbol{\alpha})} 
    = \frac{f(Y|\boldsymbol{\Theta}, \boldsymbol{\alpha})\pi(\boldsymbol{\Theta}|\boldsymbol{\alpha})\pi(\boldsymbol{\alpha})\pi(k)}{p(Y|\boldsymbol{\Theta}, \boldsymbol{\alpha})\pi(\boldsymbol{\Theta}|\boldsymbol{\alpha}) \pi(\boldsymbol{\alpha})} 
    = \frac{f(Y|\boldsymbol{\Theta}, \boldsymbol{\alpha})\pi(k)}{p(Y|\boldsymbol{\Theta}, \boldsymbol{\alpha})},
\end{equation}

where $p(Y|\boldsymbol{\Theta}, \boldsymbol{\alpha})$ is the marginal likelihood of the observed data. Since the marginal likelihood does not depend on $k$, we see that in terms of $k$, the posterior has the desired form. A similar calculation holds for multiple CP detection.

This formulation is appealing because it provides a natural way to quantify uncertainty in CP locations. Its main drawback is computational: in multiple CP problems, the marginal likelihood is often unavailable in closed form under non-conjugate models, so full posterior inference typically relies on MCMC or related simulation-based methods \citep{stephens1994bayesian, gelman1995bayesian}, which do not scale well. More restrictive choices, such as conjugate priors or simple location priors, can improve tractability, but at the cost of flexibility and generality.
For CP estimation alone, however, this difficulty can be sidestepped by appealing to maximum a posteriori (MAP) estimation, which is equivalent to minimizing a negative log-likelihood plus a negative log-prior term:
\[
\hat{\boldsymbol{\tau}}_{\mathrm{MAP}}
= \arg\min_{\boldsymbol{\tau}}
\left\{
-\log P(Y \mid \boldsymbol{\tau}) - \log \pi_1(|\boldsymbol{\tau}|)
-\log \pi_2(\boldsymbol{\tau} \mid |\boldsymbol{\tau}|)
\right\}.
\]

Thus, MAP estimation can be viewed as penalized likelihood, with one component governing the number of CPs and another governing their locations. This connection motivates our approach: rather than performing full Bayesian inference, we construct a prior-informed penalty that remains compatible with efficient multiple CP detection.

\subsection{Constructing a prior-informed penalty for MCD problems}
\label{sec: prior construction}
We now construct a prior-informed penalty motivated by penalized likelihood formulation of the MCD problem, which becomes minimizing: 
 \[
\sum_{j=1}^{m+1} C\!\left(y_{(\tau_{j-1}+1):\tau_j}\right) + \text{penalty}
 \]

\begin{equation}
\label{equ: penalized MCD}
    \arg\min_{\boldsymbol{\tau}} \sum_{i = 1}^{m + 1} \Big[C(y_{((\tau_{i - 1}) + 1): \tau_i}) \Big] + \text{penalty},
\end{equation}

When the cost $C$ is the negative log-likelihood (NLL), this criterion is equivalent, up to an additive term, to maximizing a posterior objective in which the penalty plays the role of a negative log-prior on CP locations. This motivates specifying the penalty directly from prior information on plausible CP locations. Once defined, the resulting penalty function is deterministic for a given time series.

To encode prior support around a prespecified center $c_i$, we define a Gaussian kernel:
\begin{equation}
\kappa_i(t) = \exp\!\left\{-\frac{(t- c_i)^2}{2\sigma^2}\right\}, \qquad t=1,\ldots,n,
\end{equation}

where $\sigma$ controls the temporal spread of prior influence. Each kernel is largest near $c_i$ and decays smoothly with distance. For multiple prior-supported center locations $c_1,c_2,\dots,c_m$, we combine these local supports through the concatenation rule
\[
1-\prod_{i=1}^m \bigl(1-\kappa_i(t)\bigr)
\]
This construction is motivated by the noisy-OR idea from probabilistic graphical models \citep{koller2009probabilistic}, where multiple possible causes are combined under a disjunctive interpretation: the overall signal is strong when at least one cause is active. In our setting, each kernel center represents one possible source of prior support for a CP, so the above quantity is large whenever $t$ lies close to at least one supported location. Since our goal is to construct a penalty rather than a support score, we use its complement $S(t)$ instead:

\begin{equation}
S(t) = \prod_{i=1}^m \bigl(1-\kappa_i(t)\bigr).
\end{equation}

Thus, $S(t)\in[0,1]$, with values near 0 at locations favored by the prior and values near 1 far from all supported centers. Furthermore, it accommodates any number and configuration of prior centers without imposing additional spacing constraints. Lastly, this function can be linearly rescaled to define a time-varying penalty taking values in $[\beta, \lambda + \beta]$: 
\begin{equation}
\label{eq: gt}
g(t) = \lambda S(t) + \beta
\end{equation}

\subsection{Exactness and pruning under a time-varying penalty in Pi-Change}
\label{sec: exactness-pruning}
The penalty construction above yields a prior-informed segmentation criterion. We show in Appendix that, when the penalty is attached to candidate CP locations, the resulting objective remains exactly solvable by dynamic programming and preserves the pruning structure of PELT. Thus, Pi-Change retains the computational appeal of PELT while allowing the penalty to vary over time according to prior support.

The key reason is that $g(t) = \lambda S(t) + \beta$ depends only on the candidate CP location $t$; it does not depend on future observations, segment parameters, or interactions among CPs. As a result, the location-specific penalty enters the recursion as an additive offset and does not alter the logic of the pruning argument. Substituting \eqref{eq: gt} into the pruning condition \eqref{eq: prune-condition} yields 
\[
F(t)+C\bigl(y_{(t+1):s}\bigr)+K+\lambda\{S(t)-S(s)\}\ge F(s),
\]

This means compared with the standard PELT, only the penalty difference affects comparisons between competing candidate CPs. Thus, locations receiving weaker prior support are more likely to be pruned, while prior-favored locations remain competitive. Note that Pi-Change does not lower the global penalty level relative to standard PELT. Furthermore, the objective function \eqref{eq: pi-change-obj} can be seen to retain the usual global penalty on the number of CPs through the term $m\beta$, while adding a location-specific penalty determined by prior support. Pi-Change thus combines a global sparsity penalty on CP count with a location-specific penalty that discourages CPs in regions unsupported by prior information.

$\lambda$ also controls the extent to which prior information tilts comparisons between candidate CPs: larger $\lambda$ make pruning more likely for locations with larger $S(t)$, that is, locations less supported by prior information. In practice, we recommend choosing $\lambda$ so that the maximum uplift due to $\lambda$ is of the same order as the baseline penalty $\beta$. This allows prior information to influence CP selection without overwhelming the contribution of the segment cost. Under this formulation, Pi-Change does not reduce the global level of penalization relative to PELT. Instead, it selectively increases the penalty away from prior-favored locations, thereby discouraging CPs unsupported by external information.

\section{Simulation Studies}
In this section, we evaluate Pi-Change in a controlled setting designed to mimic week-long actigraphy data. These simulations assess whether Pi-Change improves estimation of sleep- and wake-onset times (SWOTs) by discouraging false positive CP estimates, and whether the method remains robust when prior centers are placed far from the true CPs. 

\subsection{Simulation Setup}
Actigraphy often exhibits excess zeros during inactive periods, especially sleep, together with positive activity values that are continuous and right-skewed \citep{chen2024validating}. We therefore simulate actigraphy sequences using a zero-augmented Gamma (ZAG) distribution, which captures these features by assigning positive probability to zero while modeling nonzero observations with a Gamma distribution. Accordingly, each observation in the simulated time series is generated as follows:

\begin{align}
\label{equ: ZAG pdf}
     f_{ZAG}(y_i|\boldsymbol{\eta}_i) = p_i I(y_i = 0) + (1 - p_i) I(y_i > 0) \frac{1}{\Gamma(\xi) \theta_i^\xi} y_i^{\xi - 1}\exp\left(-\frac{y_i}{\theta_i}\right), \\
     \boldsymbol{\eta}_i = (\xi, \theta_i, p_i),\quad y_i \ge 0, \quad \xi, \theta_i > 0, \quad 0 \le p_i < 1.\notag 
\end{align}

Here, we assume the shape parameter $\xi$ is known and constant throughout the time series, but the scale parameter $\theta_i$ and probability of observing a zero $p_i$ are unknown and  vary between segments. 

We simulated 3000 time series, each of which is composed of $7$ sleep-wake cycles. Each sleep-wake cycle was generated by three types of segments $S_1$, $S_2$, $S_3$, by sampling from the following simulation settings: 
\begin{itemize}

\item Segment lengths $N_1$, $N_2$, and $N_3$, where
\[
N_j = 240 + 10K_j,\qquad K_j \overset{\text{i.i.d.}}{\sim} \mathrm{Uniform}\{0,1,\ldots,24\}, \quad j=1,2,3.
\] so each segment lasted between 4 and 8 hours.

\item $S_1$ and $S_2$ represent diurnal activity and $S_3$ represents nocturnal activity, with zero-probabilities sampled as  \[
\begin{aligned}
    p_1, p_2 &\overset{\text{i.i.d.}}{\sim} \mathrm{Uniform}\{0.05,0.06,\ldots,0.45\}, \\
    p_3 &\overset{\text{i.i.d.}}{\sim} \mathrm{Uniform}\{0.55,0.56,\ldots,0.95\},
\end{aligned}
\]

\item All segments in a time series share a unified shape sampled as \[
    \xi \overset{\text{i.i.d.}}{\sim} \mathrm{Uniform}\{0.50,0.51,\ldots,1.30\}.
\]
\end{itemize}

To assess performance across varying levels of signal strength, we consider three contrast scenarios that differ in the separation between daytime and nighttime scale parameters:

\begin{itemize}
    \item \textbf{Low contrast:} $\theta_D \in \{110, 111, \ldots, 170\}$, $\theta_N \in \{85, 86, \ldots, 125\}$
    \item \textbf{Moderate contrast:} $\theta_D \in \{190, 191, \ldots, 280\}$, $\theta_N \in \{55, 56, \ldots, 80\}$
    \item \textbf{High contrast:} $\theta_D \in \{320, 321, \ldots, 450\}$, $\theta_N \in \{35, 36, \ldots, 50\}$
\end{itemize}

By generating three types of segments at random with replacement, this construction generated sequences consisting of 7 concatenated sleep-wake cycles. Within a cycle, the transition from $S_1$ to $S_2$ represents a within-day CP irrelevant to SWOTs, the transition from $S_2$ to $S_3$ represents SOT, and the transition from $S_3$ to the next cycle represents WOT. Because the series ends at the end of the 7th nocturnal segment, the final wake-onset transition was not observed. Consequently, each sequence contained 20 CPs in total: 7 within-day CPs that should ideally be pruned, 7 sleep-onset CPs, and 6 wake-onset CPs, yielding 13,000 true CPs (SWOTs) and 7,000 irrelevant CPs across 1,000 simulations per contrast level. 

We classify CP estimates into four types, using a 90-minute matching window throughout. A true target CP is considered detected if at least one estimated CP falls within 90 minutes of it; otherwise it is a false negative. Among detections that do fall within a target window, any beyond the first are counted as clustered false positives. Unmatched detections falling within 90 minutes of a within-day CP are considered irrelevant false positives, and unmatched detections falling near neither target nor within-day CPs are noted as stray false positives.

We next consider two versions of Pi-Change algorithm to assess its robustness to prior mis-specification: one with penalty centers close to the true CPs, allowing random errors of less than 10 minutes, and one with centers shifted 60 -- 120 minutes away from the truth. Throughout the simulation studies, we compare these two versions of Pi-Change with standard PELT. The detection accuracy is evaluated by 1) the total number false positives and false negatives, and 2) the distance between detected true CPs and their estimated CPs, using the mean distance when clustered false positives are present.  

\begin{figure}[htbp]
	\includegraphics[width= 0.8\linewidth]{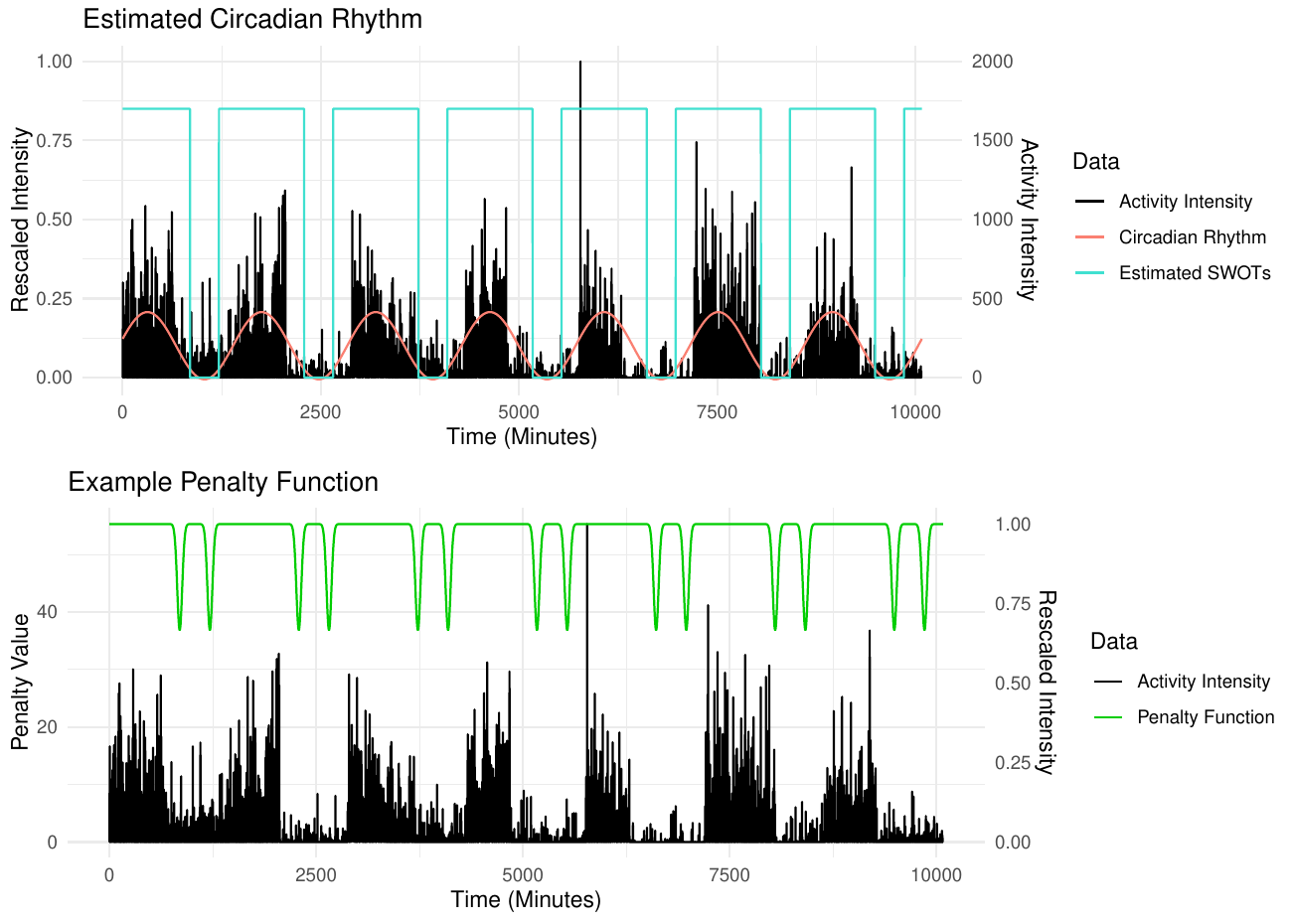}
	\caption{Plots of rescaled actigraphy data, the estimation of the circadian rhythm and the associated penalty function.}
    \label{fig: pen viz}
\end{figure}

Figure \ref{fig: pen viz} illustrates the proposed penalty construction for this simulation study. Here, the circadian rhythm is represented by the best-fitting cosinor curve, and the penalty centers are determined by the dichotomized cosinor curve \citep{chen2024validating}. The procedure described in Section \ref{sec: prior construction} then yields a time-varying penalty with a prior spread of 30 minutes. To compare PELT and Pi-Change as fairly as possible, the minimum value that Pi-Change's penalty is kept at $(|\boldsymbol{\Theta}| + 2) * \log(N)$ at all times, equivalent to the constant Modified BIC penalty in PELT \citep{zhang2007modified}. In addition, PELT uses our ZAG-based cost function, so that only the penalty methods contribute to observed simulation differences. For consistency, the prior spread $\sigma = 30$ minutes for every time series penalty.

\subsection{Simulation Results}

\begin{figure}[htbp]
	\includegraphics[width=0.7\linewidth]{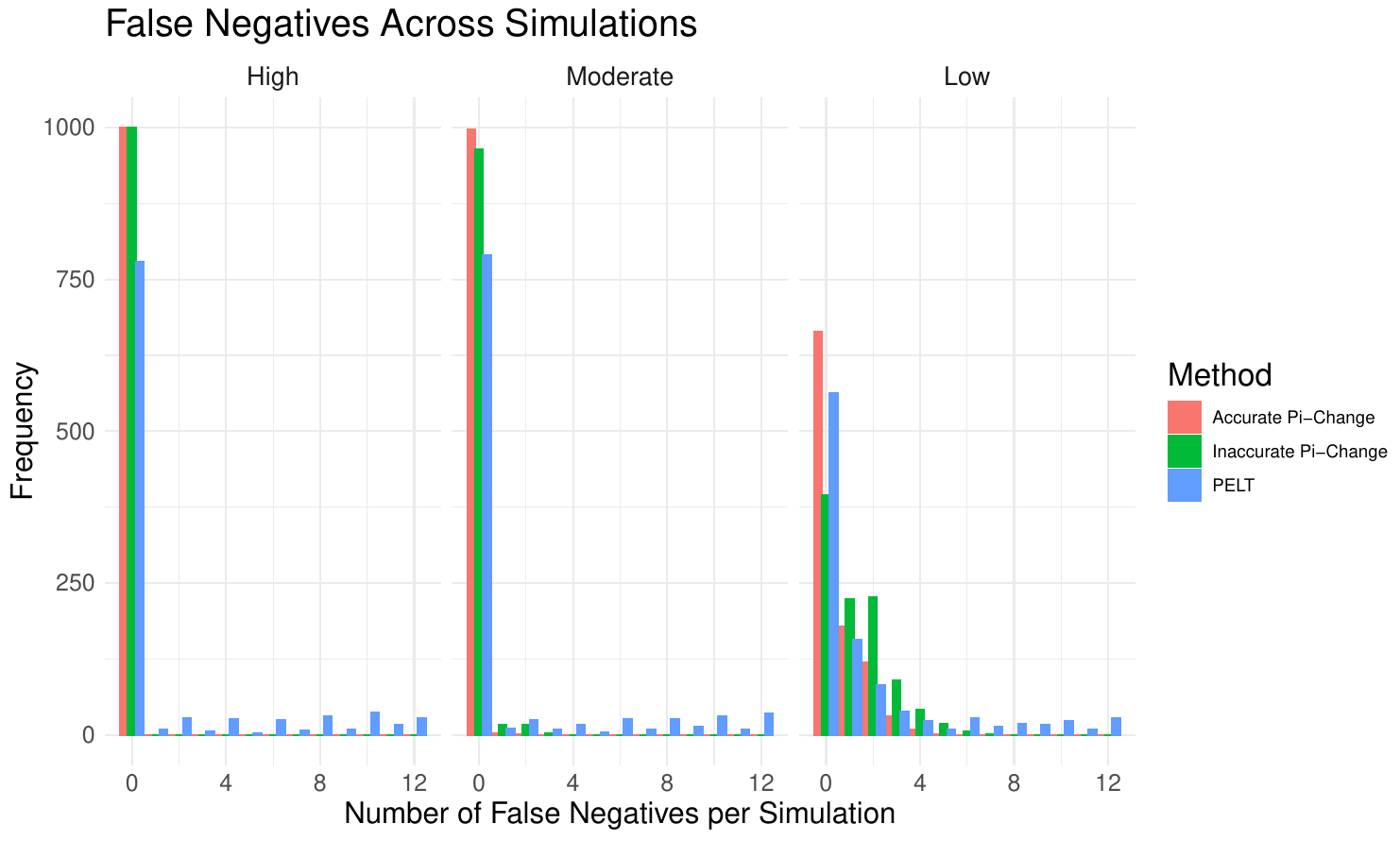}
	\caption{A histogram of the counts of false negatives estimated by PELT and Pi-Change. A false negative is a SWOT that has no CP estimate within 90 minutes of its location. Each chart corresponds to a contrast level specifying the average difference in segment scale parameters.}
	\label{fig: false neg}
\end{figure}

Figure \ref{fig: false neg} displays false negative counts across methods and contrast levels. Under high and moderate contrast, both Pi-Change methods identify nearly all desired CPs, with distributions tightly concentrated at zero. In contrast, PELT failed to identify desired CPs in about 25\% of the time series. Under low contrast, accurate Pi-Change continues to outperform PELT in false negative rate, while inaccurate Pi-Change is more visibly harmed by prior mis-specification. Nevertheless, even under low contrast with displaced prior centers, Pi-Change avoids the most severe failure mode observed in PELT, where a substantial fraction of true CPs are missed entirely.

\begin{figure}[htbp]
	\includegraphics[width=0.7\linewidth]{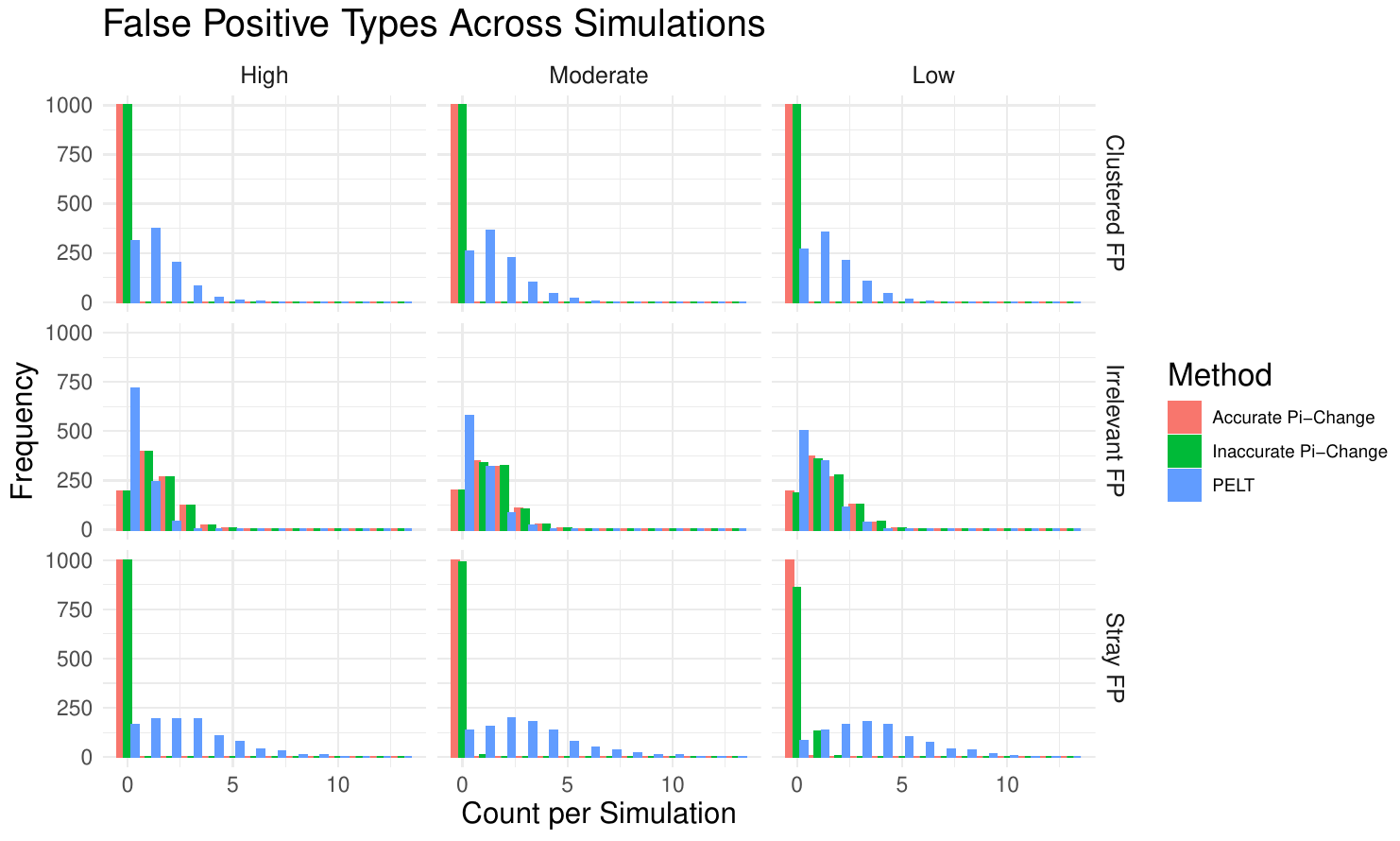}
	\caption{A panel of histograms of the counts of types of false positives estimated by PELT and Pi-Change, organized by false positive type and contrast level. Clustered CPs are additional change points found within a target window of 90 minutes. Irrelevant CPs are those found within the target window of within-day CPs rather than the SWOTs. Stray CPs are those estimates which are not in any target window.}
	\label{fig: false pos}
\end{figure}

Figure \ref{fig: false pos} displays false positive counts by type. Clustered and stray false positives are virtually absent for both Pi-Change variants across all contrast levels, whereas on average PELT produces 1 to 3 per simulation (Table \ref{tab: fp fn avgs}). In contrast, PELT detects fewer irrelevant CPs compared to Pi-Change. However, the number of additional irrelevant false positives Pi-Change estimates is overshadowed by the number of clustered and stray false positives that PELT estimates, and this holds across different contrast levels.

\begin{table}[htbp]
\centering
\resizebox{\textwidth}{!}{%
\begin{tabular}{||r c c c c c||}
\hline
\textbf{Method} & \textbf{Contrast} & \textbf{False Negative} & \textbf{Clustered FP} & \textbf{Irrelevant FP} & \textbf{Stray FP} \\ [0.5ex]
\hline
Pi-Change (Accurate)   & Low      & 0.546 & 0.000 & 1.480 & 0.003 \\
Pi-Change (Inaccurate) & Low      & 1.240 & 0.000 & 1.510 & 0.146 \\
PELT                   & Low      & 1.740 & 1.360 & 0.697 & 3.390 \\
\hline
Pi-Change (Accurate)   & Moderate & 0.004 & 0.000 & 1.430 & 0.001 \\
Pi-Change (Inaccurate) & Moderate & 0.059 & 0.000 & 1.440 & 0.012 \\
PELT                   & Moderate & 1.530 & 1.370 & 0.541 & 2.890 \\
\hline
Pi-Change (Accurate)   & High     & 0.000 & 0.000 & 1.400 & 0.000 \\
Pi-Change (Inaccurate) & High     & 0.000 & 0.000 & 1.390 & 0.001 \\
PELT                   & High     & 1.590 & 1.170 & 0.326 & 2.530 \\
[0.5ex]
\hline
\end{tabular}}
\caption{A table presenting the average number of false negative and false positive CP detections per simulated time series, organized by type. Clustered CPs are additional CPs found within a target window of 90 minutes. Irrelevant CPs are those found within the target window of within-day CPs rather than the SWOTs. Stray CPs are those estimates which are not in any target window.}
\label{tab: fp fn avgs}
\end{table}

Table \ref{tab: false pos error quants} and Figure \ref{fig: errors zoom} examine estimation errors for well-detected CPs, measured as the distance between each true target CP and its closest estimate within the matching window. Both versions of Pi-Change achieve substantially smaller mean absolute errors than PELT across all contrast levels. Under low contrast, accurate and inaccurate Pi-Change yield MAEs of 2.3 and 5.0 minutes respectively, compared to 13.2 minutes for PELT. Under moderate contrast, MAEs fall to 1.6 and 2.2 minutes for Pi-Change against 7.2 minutes for PELT. Under high contrast all methods improve considerably, though Pi-Change retains its advantage with MAEs near 1.0 and 1.1 minutes versus 5.3 minutes for PELT. This contrast is most pronounced in the upper tail: under low contrast, PELT errors reach above 30 minutes at the 90th percentile, while both Pi-Change variants remain below 10 minutes. As contrast increases, all three methods concentrate near zero, but Pi-Change consistently maintains a sharper peak and lighter tail than PELT. Together, these results indicate that the Pi-Change penalty not only reduces false detections, but also offers much tighter control over large CP estimation errors, with the advantage most pronounced when the contrast between diurnal and nocturnal activity is low, or in other words, when the effect size of parameter differences is low.

\begin{table}[htbp]
\centering
\resizebox{\textwidth}{!}{%
\begin{tabular}{||r c c c c c c c c c||}
\hline
\textbf{Method} & \textbf{Contrast} & \textbf{MAE} & \textbf{Min} & \textbf{25\%} & \textbf{50\%} & \textbf{75\%} & \textbf{90\%} & \textbf{95\%} & \textbf{Max} \\ [0.5ex]
\hline
Pi-Change (Accurate)   & Low      &  2.250 & 0 & 0 & 1 &  3 &  6.0 &  9.00 & 72 \\
Pi-Change (Inaccurate) & Low      &  5.043 & 0 & 0 & 1 &  3 & 10.0 & 27.00 & 90 \\
PELT                   & Low      & 13.241 & 0 & 1 & 5 & 21 & 39.5 & 47.00 & 90 \\
\hline
Pi-Change (Accurate)   & Moderate &  1.584 & 0 & 0 & 1 &  2 &  4.0 &  7.00 & 74 \\
Pi-Change (Inaccurate) & Moderate &  2.219 & 0 & 0 & 1 &  2 &  5.0 &  8.00 & 90 \\
PELT                   & Moderate &  7.213 & 0 & 0 & 1 &  6 & 32.0 & 39.00 & 90 \\
\hline
Pi-Change (Accurate)   & High     &  0.999 & 0 & 0 & 0 &  1 &  3.0 &  4.00 & 17 \\
Pi-Change (Inaccurate) & High     &  1.063 & 0 & 0 & 0 &  1 &  3.0 &  5.00 & 57 \\
PELT                   & High     &  5.338 & 0 & 0 & 1 &  3 & 30.0 & 37.00 & 90 \\
[0.5ex]
\hline
\end{tabular}}
\caption{Error comparison between PELT and Pi-Change (with accurate and inaccurate prior centers). Errors are measured as the distance between a desired SWOT CP and its closest CP estimates. By definition of detection window, the maximum error is 90 minutes. MAE: mean absolute error. Min: minimum. Max: Maximum. 25\%, 50\%, 75\%, 90\%, 95\% are quantiles of the error distributions.}
\label{tab: false pos error quants}
\end{table}

\begin{figure}[htbp]
	\includegraphics[width=0.7\linewidth]{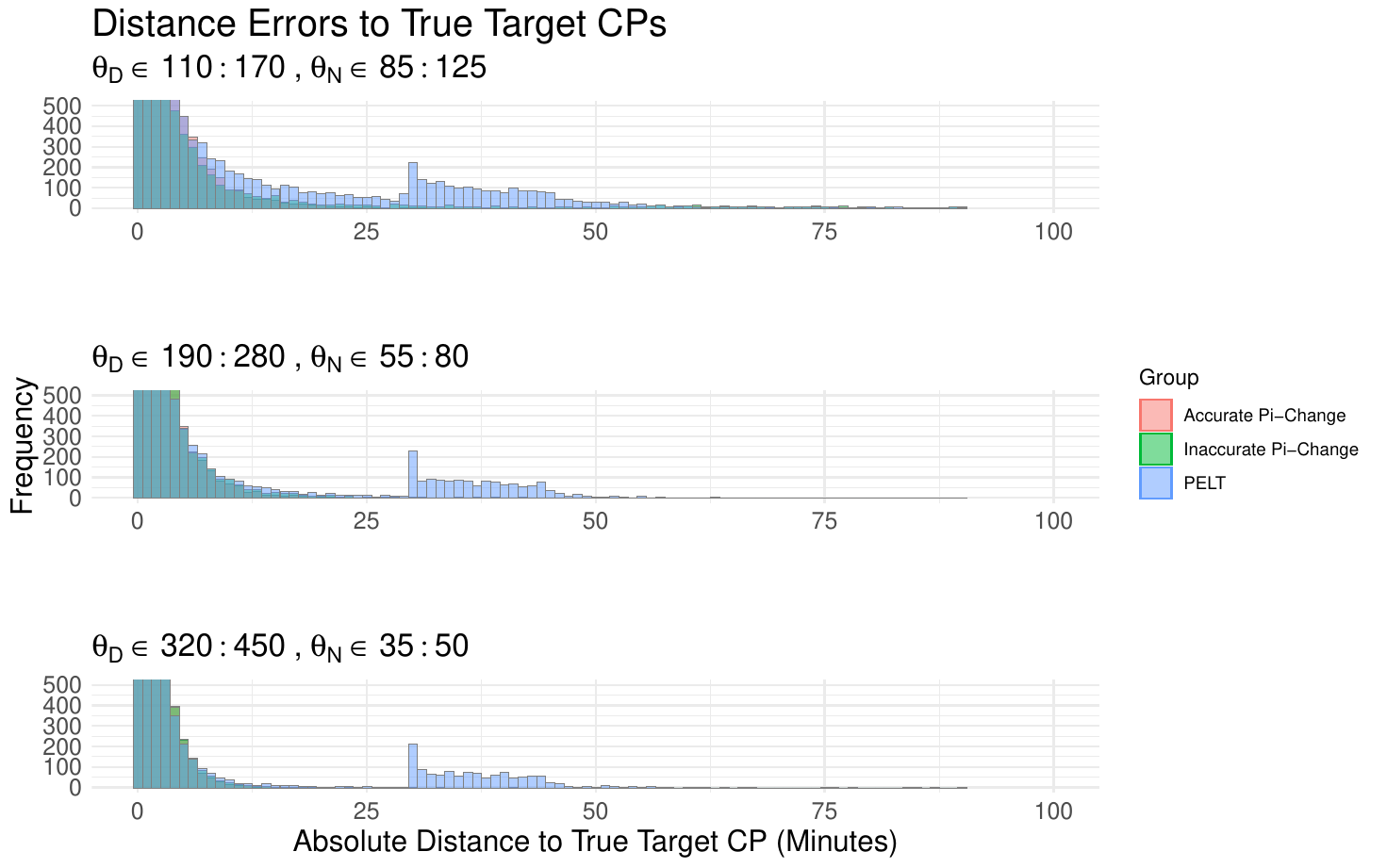}
	\caption{A zoomed-in histogram of the CP estimation errors for PELT and Pi-Change. Errors are defined as the absolute difference between each true CP and the mean of the clustered CP estimates within a 90-minute window of the truth. Centered on frequencies less than 500 to observe details in the histogram tails. Contrast levels from top to bottom: Low, Moderate, High.}
	\label{fig: errors zoom}
    
\end{figure}

\section{Applications}

In this section, we apply the Pi-Change algorithm to several time series, to examine its empirical performance.
\subsection{World War–Informed CP Detection in Historical Combat Deaths}
We first apply Pi-Change to historical combat-death data, where major wars provide natural prior information on plausible CP locations. This example allows us to assess whether historically important events align with meaningful structural changes.
\begin{figure}[htbp]
	\includegraphics[width=0.8\linewidth]{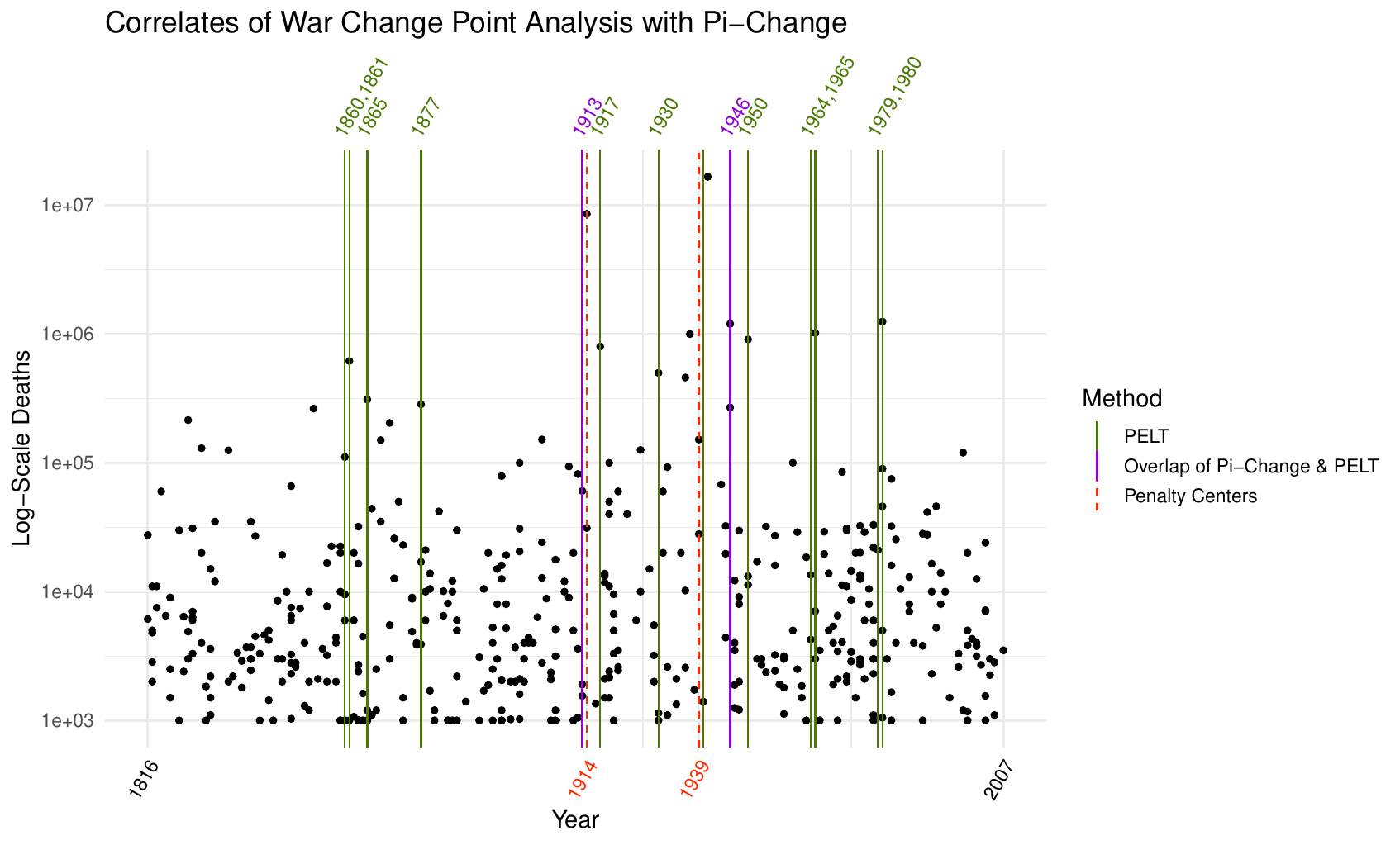}
	\caption{A plot of log-scaled combat deaths by start dates of conflicts. CPs detected by only PELT are indicated by green lines, while those detected by both Pi-Change and PELT are indicated by purple lines. Labels at the top indicate the year of the detected CP, colored by method. Penalty center locations are indicated by red dashed lines, with labels at the bottom of the plot.}
	\label{fig: combat deaths}
    
\end{figure}
The Correlates of War (COW) 3.2 dataset collects the combat deaths across inter-state, intra-state, extra-state, and non-state conflicts from 1816 to 2007. Using this data, Fagan et al \cite{fagan2020change} studied distributional changes in the battle deaths over time. To accommodate the heavy-tailed nature of battle deaths, they extended PELT using multiple power-law models.. 

We revisit this question using Pi-Change and ask whether World War I and World War II mark CPs in combat-death rates.  Unlike Fagan et al, who deliberately avoided any prior information, we treat the two world wars as historically well-motivated candidate CP locations and examine whether additional CPs remain. Thus, we set two priors for CPs: 1914 (the beginning of World War I), and 1939 (the beginning of World War II). The prior spread is set to 10. 

To ensure compatibility, we replicated Fagan et al's data-cleaning process. Each combat event was indexed by its start date to reduce autocorrelation from deaths within the same conflict episode. We excluded events with unknown death counts, either due to incomplete reporting or overlap with concurrent events such as famines, and restricted the analysis to events with at least 1,000 combat deaths. Data were analyzed by Pi-Change under the ZAG model, with a minimum segment length for detection of 10 years. 

Figure  \ref{fig: combat deaths} shows that Pi-Change recovers the two dominant historical breaks while eliminating several additional CPs selected by unweighted PELT. Importantly, the estimated CPs occur near --- but not exactly at --- the prior centers: Pi-Change detected breaks at 1913 and 1946, slightly away from penalty centers at 1914 and 1939. This is desirable, because it shows that the prior-informed penalty acts as soft temporal support rather than as a hard anchor. The resulting segmentation lent credence to the “long peace” or “great violence” hypothesis in political science, which posits a decline in the lethality of large-scale conflicts over the past 80 years. In this application, Pi-Change therefore preserves historically meaningful large-scale regime shifts while allowing the data to determine their precise timing, which is consistent with the goal of using contextual information to guide, rather than dictate, CP estimation.

\subsection{2003 - 2009 Oil Price Bubble}
\begin{figure}[htbp]
	\includegraphics[width=0.85\linewidth]{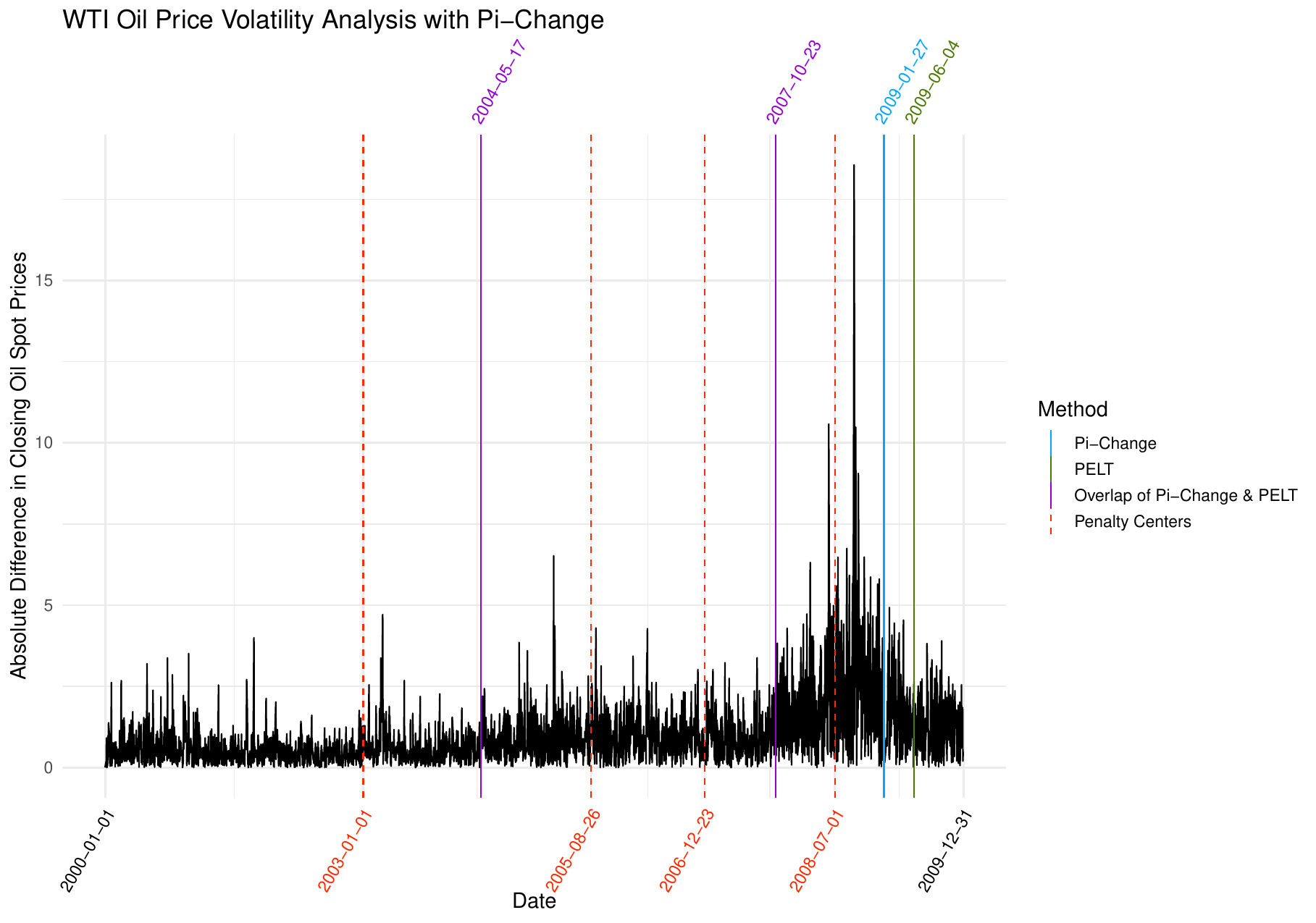}
	\caption{A plot of the absolute difference between closing prices of WTI oil in the 2000's. CPs detected by Pi-Change are indicated by blue lines, while those indicated by PELT are indicated by green lines. The CPs identified by both methods are indicated by purple lines. Labels for the dates of the CPs are given above the graph, colored by method. The locations of the penalty centers are indicated by dashed red lines, with date labels below the plot.}
	\label{fig: oil bubble}
\end{figure}

In the 2000's, inflation-adjusted oil prices rose sharply before collapsing by the end of 2008. Because the drivers of this boom-and-bust period have been and remain debated \citep{nadal2017time, cifarelli2021navigating}, we use Pi-Change to ask whether major geopolitical and macroeconomic events align with short-term CPs in the price volatility in West Texas Intermediate's (WTI) crude-oil price. Specifically, we analyze daily WTI spot prices in Cushing, Oklahoma, from January 1st, 2000 to December 31st, 2009, using the absolute daily change in closing price as the time series of interest. For this application, we considered ZAG-distribution cost function in Pi-Change, with penalty centers corresponding to the following events: January 1st, 2003 (the year the bubble began), August 26th, 2005 (Hurricane Katrina landfall), December 23rd, 2006 (Iran Nuclear Sanctions), and July 1st, 2008 (the approximate beginning of the great recession). The prior spread was set to 130, equivalent to two financial quarters (26 weeks * 5 days / week), and the minimum segment length for detection was set to 65 days, or one financial quarter. 

Figure \ref{fig: oil bubble} suggests that the dominant CPs are driven primarily by the data rather than by specified prior events. Both PELT and Pi-Change identify CPs on May 17, 2004 and October 23, 2007, indicating that these dates correspond to strong shifts in the volatility series. The main difference between Pi-Change and PELT is at the end of the elevated-volatility period, which Pi-Change places at January 27, 2009 and PELT at June 4, 2009. Notably, none of the Pi-Change estimates lies close to a penalty center; the nearest is still about six months away, exceeding the 130-day prior spread. This suggests that the selected geopolitical events exerted limited influence on the shifts in the oil price volatility. This application shows that when hypothesized events do not drive major changes in the data, Pi-Change remains robust to prior mis-specification and still recovers the main shifts independently of those events.

\subsection{2020 - 2022 COVID Mobility Changes in the U.S. and UK}

\begin{figure}[htbp]
	\includegraphics[width= \linewidth]{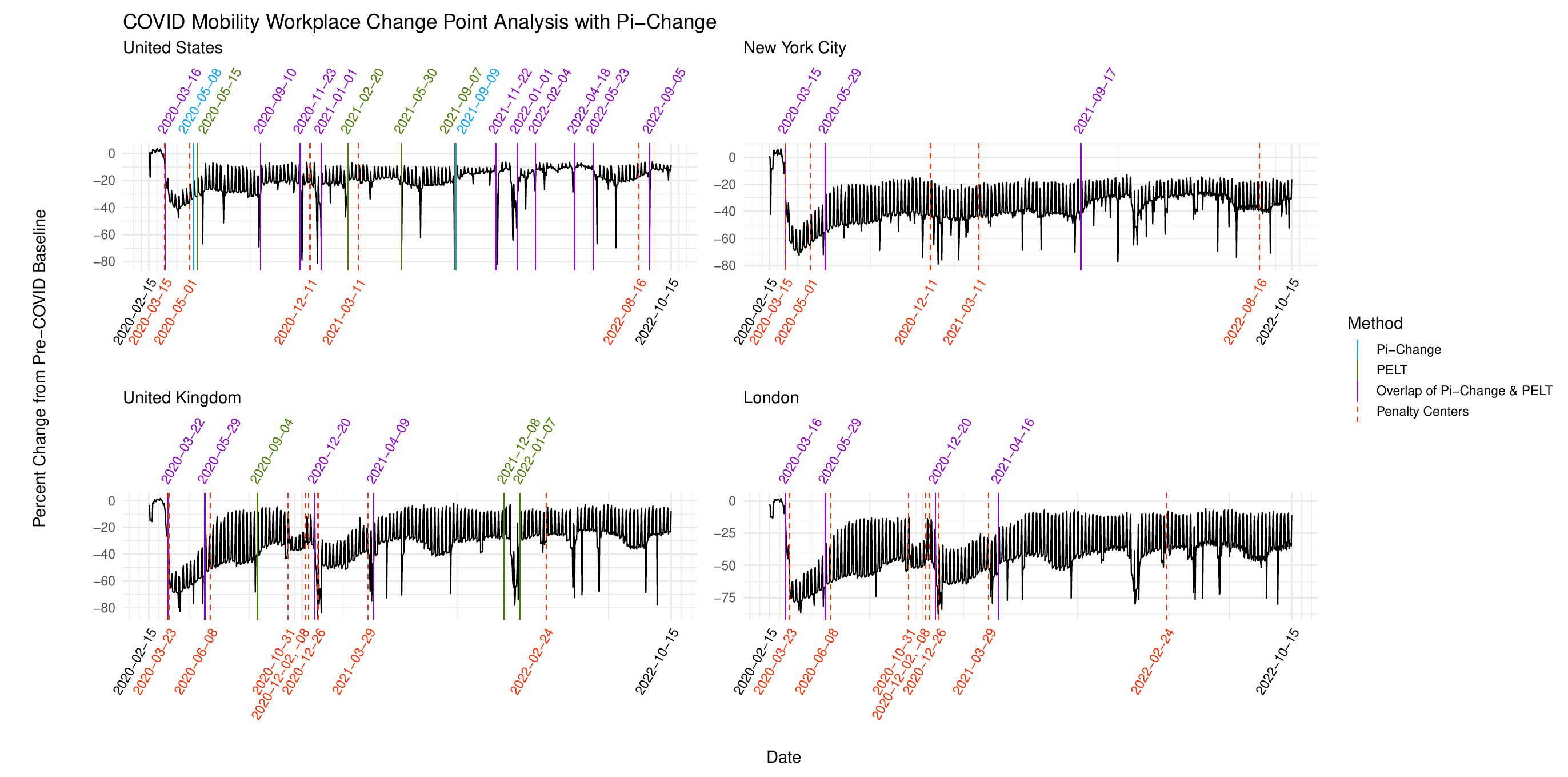}
	\caption{A plot of the percent change in workplace traffic data from 2020 to 2022 in the United States, New York City, the United Kingdom, and London. CPs detected by Pi-Change are indicated by blue lines, while those indicated by PELT are indicated by green lines. The CPs identified by both methods are indicated by purple lines. Labels for the dates of the CPs are given above each graph, colored by method. The locations of the penalty centers are indicated by dashed red lines, with date labels below the plot.}
	\label{fig: covid mobility}
    
\end{figure}
The COVID-19 pandemic disrupted daily life worldwide, where major policy shifts led to abrupt changes in human mobility, providing natural context for plausible CP locations. To track these changes, Google released a large-scale mobility dataset reporting percentage changes in visits and time spent across several location categories, including retail and recreation, grocery and pharmacy, parks, transit stations, workplaces, and residential areas. We showcase how Pi-Change can be used to detect the changes in mobility given federal policy enactment in the United States (U.S.) and United Kingdom (U.K.).

We extracted country-level workplace mobility series for the U.S., the U.K., and city-level workplace mobility series for New York City (NYC), and London, measured as percentage changes relative to a pre-COVID baseline. The city level information was defined as the five counties in New York City proper and the 32 boroughs and city of London that make up Greater London. The country-level data was averaged each day over the 50 U.S. States and the entire U.K.. For this application, segment costs were computed as NLL under a Gaussian model. In Pi-Change, prior information was incorporated through event-centered penalties reflecting major federal policy and economic interventions plausibly associated with changes in mobility. For the U.S., the penalty centers were set at March 15, 2020 (approximate onset of state-level lockdowns), May 1, 2020 (approximate onset of reopening), December 11th, 2020 (Pfizer-BioNTech vaccine emergency use authorization), March 11, 2021 (American Rescue Plan Act), and August 16, 2022 (Inflation Reduction Act). For the U.K., the penalty centers were set at March 23rd, 2020 (beginning of Stay At Home order), June 8th, 2020 (approximate end of lockdown), October 31st, 2020 (onset of second lockdown), December 2nd, 2020 (end of second lockdown), December 8th, 2020 (onset of Pfizer-BioNTech vaccine use), December 26th, 2020 (onset of third lockdown), March 29th, 2021 (end of third lockdown), and February 24th, 2022 (last restrictions due to COVID removed). For Pi-Change, we set both the prior spread and the minimum segment length to 30, reflecting the expectation that meaningful mobility changes would persist for at least several weeks.

Figure \ref{fig: covid mobility} shows that both Pi-Change and PELT detect the major early-pandemic shifts in workplace mobility, including the sharp decline at the onset of the initial lockdown period and the subsequent partial reopening. The main difference between methods arise in the aggregate national series: PELT selects several extra CPs in comparison with Pi-Change, suggesting that Pi-Change prunes shorter-lived fluctuations while preserving the dominant policy-related regime shifts. The figure also reveals an important between-country contrast. While London and New York City follow broadly similar trajectories, with a sharp initial collapse, partial recovery, and persistent depression relative to baseline, the national series diverge more substantially: London closely tracks the U.K., whereas New York City is much simpler than the aggregate U.S. series, consistent with greater spatial heterogeneity in the U.S. response. A renewed late-2020 dip is evident in the U.K. and London series, plausibly reflecting both holiday-season behavior and renewed winter restrictions, while later breaks in the U.S. national series around late 2021 to early 2022 are more suggestive of holiday-related disruptions than of a single nationwide policy shift. Another notable feature is the stronger weekly oscillation in the U.K., London, and New York City series, but much less so in the aggregate U.S. series, likely because national averaging across heterogeneous regions dampens local weekday--weekend commuting patterns. Overall, the figure suggests that Pi-Change is particularly useful for separating broad policy-driven mobility changes from additional breaks induced by short-term seasonal effects or cross-region heterogeneity.

\section{Discussion}
We have developed a method for incorporating prior information on multiple CP locations and, to our knowledge, is the first general algorithm for doing so in MCD problems. Building on the fast and exact PELT algorithm, Pi-Change uses prior information to guide detection toward scientifically meaningful locations. The method constructs penalties from prior distributions so that the resulting penalized likelihood admits an interpretation analogous to maximizing a posterior quantity. This penalization allows Pi-Change to distinguish CPs arising from heterogeneous mechanisms and prune spurious CPs, while retaining the computational efficiency of PELT.

Our simulation studies and three applications highlight the advantages of Pi-Change as a prior-informed MCD method. In actigraphy-based simulations, Pi-Change avoided estimating CP locations that were unsupported by prior information, reduced estimation error, and remained robust to biased prior centers. In the real-data examples, it used historical, geopolitical, and policy context to guide detection toward CPs of interest. In the Correlates of War data, it recovered the historically meaningful CPs identified by \cite{fagan2020change} when World War I and World War II were used as penalty centers. In the oil-price volatility data from the 2003--2009 bubble, Pi-Change selected CPs similar to those from PELT even when the proposed geopolitical events were not closely aligned with the detected changes. In the COVID-19 mobility data, Pi-Change captured greater CP heterogeneity in the United States and relative homogeneity in the United Kingdom.

Currently, work on scalable Bayesian MCD methods remains sparse. Although some Bayesian methods have been proposed, they did not focus on incorporating prior information on CP locations. For example, \cite{barry1993bayesian} introduced product-partition formulations, and \cite{fearnhead2006exact} extended their work and proposed exact Bayesian recursions with priors on CP number and positions. However, these priors typically act through partition structure, segment lengths, or gap distributions, rather than through externally specified, nonstationary prior support on CP locations. Similarly, Bayesian online run-length methods \citep{adams2007bayesian} are designed to infer the posterior distribution of the most recent run length via a hazard or gap model, rather than to encode prior plausibility of particular retrospective CP locations. The contribution of Pi-Change is thus distinct: it uses a prior-informed, location-specific penalty to encode external temporal knowledge directly in the optimisation criterion, while preserving the exact dynamic programming recursion and pruning structure needed for efficient MCD.

Pi-Change has several limitations. First, although Pi-Change incorporates prior information, it trades full Bayesian inference for computational efficiency by relying on maximum a posterior estimation, thus does not directly provide posterior uncertainty quantification on the CP estimates. In settings where uncertainty intervals are required, post-hoc inferential procedures may be applied after segmentation \citep{yau2016inference}. Second, the influence of prior information is controlled jointly by the penalty scale and kernel prior spread, so practical performance depends on sensible tuning. Developing more interpretable tuning schemes and principled sensitivity analyses is therefore an important next step. Third, the current formulation is designed for abrupt changes in piecewise-stationary sequences. When strong trend or seasonal structure is present, these components should be modeled or removed before applying the method. Extending the framework to other types of structural change, including trend and seasonal breaks, is a natural direction for future work.

The applications of Pi-Change are broad. Pi-Change is well suited to settings in which CPs may arise from heterogeneous mechanisms and external information helps identify which changes are scientifically meaningful. Examples include financial market data, where CPs may correspond to crises, policy announcements, as well as short-lived volatility shocks in returns; biomedical, behavioral, climate, energy, and industrial monitoring data, where known events or domain structure can guide detection toward relevant structural changes and away from incidental fluctuations. Another important application of Pi-Change is in estimating the temporal delay between a contextual event and the resulting structural change in the observed data, such as the lag from a policy intervention, environmental exposure, or major social disruption to its measurable effect on the underlying process.

In summary, Pi-Change provides a general framework for incorporating prior information into multiple CP detection while retaining computational efficiency. Its usage lies not simply in detecting change, but in identifying meaningful CPs while pruning spurious ones.

\textbf{Data Availability}
The Correlate of War data can be downloaded from the website of Correlates of War \citep{maoz2019dyadic}. 
The WTI oil price data can be found on the website of \cite{fred_dcoilwtico}. COVID-19 mobility data are available on \cite{kaggle_google_covid_mobility}. The simulation data are fully reproducible from the code provided. All code needed to reproduce the simulations, analyses, and figures in this paper will be made publicly available upon publication. Work is underway to release the Pi-Change algorithm as an R package on CRAN. 

\bibliographystyle{abbrvnat}
\bibliography{pi-change.bib}

\newpage
\appendix
\section{Proofs}
\label{app:proofs}
We define the Pi-Change penalty at location $t$, $g(t)$, as \eqref{eq: gt}. Here, $\beta$ is the baseline PELT penalty, $\lambda \ge 0$ controls the strength of the prior information, and $S(t)\ge 0$ is a time-varying kernel derived from the prior information. When S(t) = 0, the minimum Pi-Change penalty equals the baseline PELT penalty $\beta$, indicating that locations unsupported by prior information receive a larger penalty.

We then define the Pi-Change objective as
\begin{equation}
\label{eq: pi-change-obj}
Q(\tau)
=
\sum_{j=1}^{m}
\left\{
C\bigl(y_{(\tau_{j-1}+1):\tau_j}\bigr)+g(\tau_j)
\right\}
+
C\bigl(y_{(\tau_m+1):N}\bigr)
\end{equation}
where $C(\cdot)$ is the segment cost. Thus, a penalty is charged at each change point
location $\tau_j$, but not at the terminal endpoint $N$.

For $s\le N$, let
\[
\mathcal{T}_s
=
\left\{
\tau: 0=\tau_0<\tau_1<\cdots<\tau_m<\tau_{m+1}=s
\right\},
\]
and let $F(s)$ denote the minimum Pi-Change objective on the data $y_{1:s}$:
\[
F(s)=\min_{\tau\in\mathcal{T}_s} Q_s(\tau),
\]
where $Q_s(\tau)$ is defined analogously to \eqref{eq: pi-change-obj} with endpoint $s$.

\begin{proposition}
\label{prop: recursion}
The Pi-Change objective satisfies the recursion
\begin{equation}
\label{eq: recursion}
F(s)
=
\min_{1\le t<s}
\left\{
F(t)+C\bigl(y_{(t+1):s}\bigr)+g(t)
\right\}.
\end{equation}
Consequently, dynamic programming yields the exact global minimizer of the Pi-Change
criterion.
\end{proposition}

\begin{proof}
Consider an optimal segmentation of $y_{1:s}$, and let $t$ be its last CP prior to $s$.
Its total cost can then be decomposed into three parts: the optimal cost up to time $t$, the
cost of the final segment $y_{(t+1):s}$, and the penalty charged for placing a CP
at $t$. Hence
\[
Q_s(\tau)=Q_t(\tau_{1:t})+C\bigl(y_{(t+1):s}\bigr)+g(t),
\]
where $\tau_{1:t}$ denotes the induced segmentation on $y_{1:t}$. If the segmentation on
$y_{1:t}$ were not optimal, it could be replaced by a lower-cost segmentation, contradicting
optimality of $\tau$. Therefore the prefix cost must equal $F(t)$, and minimizing over all
possible $t<s$ yields \eqref{eq: recursion}.
\end{proof}

Proposition \ref{prop: recursion} shows that Pi-Change remains an exact optimal
partitioning algorithm. We now show that the PELT pruning argument continues to hold.

\begin{theorem}
\label{thm: pruning}
Suppose there exists a constant $K$ such that for all $t<s<T$,
\begin{equation}
\label{eq: cost-condition}
C\bigl(y_{(t+1):s}\bigr)+C\bigl(y_{(s+1):T}\bigr)+K
\le
C\bigl(y_{(t+1):T}\bigr).
\end{equation}
If, for some $t<s$,
\begin{equation}
\label{eq: prune-condition}
F(t)+C\bigl(y_{(t+1):s}\bigr)+g(t)+K
\ge
F(s)+g(s),
\end{equation}
then $t$ cannot be the optimal last change point prior to any future time $T>s$.
\end{theorem}

\begin{proof}
Fix $T>s$. Starting from \eqref{eq: prune-condition}, add
$C\bigl(y_{(s+1):T}\bigr)$ to both sides:
\[
F(t)+C\bigl(y_{(t+1):s}\bigr)+g(t)+K+C\bigl(y_{(s+1):T}\bigr)
\ge
F(s)+g(s)+C\bigl(y_{(s+1):T}\bigr).
\]
By \eqref{eq: cost-condition},
\[
C\bigl(y_{(t+1):s}\bigr)+C\bigl(y_{(s+1):T}\bigr)+K
\le
C\bigl(y_{(t+1):T}\bigr),
\]
and therefore
\[
F(t)+C\bigl(y_{(t+1):T}\bigr)+g(t)
\ge
F(s)+C\bigl(y_{(s+1):T}\bigr)+g(s).
\]
The left-hand side is the objective value obtained by taking $t$ as the last CP before $T$, while the right-hand side is the corresponding value obtained by taking $s$ as the last CP before $T$. Hence $t$ can never be preferred to $s$ for any future endpoint $T>s$, and so $t$ cannot be the optimal last CP prior to $T$.
\end{proof}
\end{document}